\documentclass[conference]{IEEEtran}

\usepackage{tikz}
\usepackage{bbm}
\usetikzlibrary{arrows}
\usetikzlibrary{calc}
\usetikzlibrary{snakes}
\usepackage{cite}
\usepackage{graphicx}
\usepackage[cmex10]{amsmath}
\usepackage{amsfonts}
\usepackage{amssymb}
\usepackage{algorithmic}
\usepackage{array}
\usepackage{mdwmath}
\usepackage{mdwtab}
\usepackage{eqparbox}
\usepackage{fixltx2e}
\usepackage{url}

\newtheorem{definition}{Definition}
\newtheorem{theorem}{Theorem}
\newtheorem{lemma}{Lemma}

\newtheorem{protocol}{Protocol}

\begin{document}

\title{The Oblivious Transfer Capacity of the Wiretapped Binary Erasure Channel}

\author{\IEEEauthorblockN{Manoj Mishra and Bikash Kumar Dey }
\IEEEauthorblockA{IIT Bombay, India\\
Email: \{mmishra,bikash\}@ee.iitb.ac.in}
\and
\IEEEauthorblockN{Vinod M. Prabhakaran}
\IEEEauthorblockA{TIFR, Mumbai, India\\
Email: vinodmp@tifr.res.in}
\and
\IEEEauthorblockN{Suhas Diggavi}
\IEEEauthorblockA{UCLA, USA\\
Email: suhas@ee.ucla.edu}}

\maketitle

\begin{abstract}
We consider oblivious transfer between Alice and Bob in the presence of an eavesdropper Eve when there is a broadcast channel from Alice to Bob and Eve. In addition to the secrecy constraints of Alice and Bob, Eve should not learn the private data of Alice and Bob. When the broadcast channel consists of two independent binary erasure channels, we derive the oblivious transfer capacity for both 2-privacy (where the eavesdropper may collude with either party) and 1-privacy (where there are no collusions).
\end{abstract}

\IEEEpeerreviewmaketitle

\section{Introduction}

The goal of secure multiparty computation (MPC) is for mutually distrusting parties to collaborate in computing functions of their data, but without revealing anything more about their data to others than what they can infer from the function outputs and data. Useful applications of secure MPC include voting, auctions and data-mining amongst several others, see e.g., \cite[Chap.
1]{CramDamNiel}. It is well known that information theoretically (unconditionally) secure computation is not possible, in general (i.e. for arbitrary functions), between two parties with noiseless communication and only common and private randomness. A combinatorial charaterization of functions that can be securely computed by two parties is given in \cite{kushilevitz1992}. Two-party secure computation, in general, requires additional stochastic resources.  Specifically, a noisy channel between the parties provides a means to achieve secure computation~\cite{CrepKilian1988}.

Oblivious Transfer (OT) has been proposed as a basic primitive (which can be derived from noisy channels) on which secure computation can be founded \cite{jkilian1988, jkilian2000}. One-out-of-two (1-of-2) string OT is a secure 2-party primitive computation, where one party, Alice, has two strings of equal lengths out of which, the other party, Bob, obtains exactly one string of his choice without Alice finding out the identity of the string selected by Bob. The (string) OT capacity of a discrete memoryless channel is the largest
string-length-per-channel-use that can be supported. OT capacity of discrete memoryless channels has been studied in~\cite{NascWinter2008,ot2007,PintoDowsMorozNasc2011}. In~\cite{NascWinter2008}, a lower bound on the string OT capacity of noisy channels and source distributions was obtained for honest-but-curious participants (i.e., the parties do not deviate from the prescribed protocol, but attempt to derive information about the other party's input that they are not allowed to know from everything they have access to at the end of the
protocol). \cite{ot2007}~characterizes the string OT capacity for generalized erasure channels, when the two parties are honest-but-curious. \cite{PintoDowsMorozNasc2011} shows that this honest-but-curious string OT capacity of generalized erasure channels can, in fact, be achieved even when the two parties are malicious. 

A natural consideration when using noisy channels is the presence of third parties who may derive useful information about the
computation. For example, consider the noisy resource as a wireless channel.  In this case, an eavesdropper who receives partial
information about the transmissions can use it to deduce the output or data of the parties. Motivated by this, we study the OT capacity of an erasure channel in the presence of an eavesdropper (Figure~\ref{fig:ot-setup}). To the best of our knowledge, this
problem has not been studied before. We limit our study to the case of honest-but-curious parties Alice and Bob and a passive eavesdropper Eve. We consider secrecy regimes where Eve may collude with Alice or Bob (2-privacy) and where there is no such collusion (1-privacy). These requirements are made more precise in the next section. We derive the 1-of-2 string OT capacity, for both 1-privacy and 2-privacy, in the setup of Figure~\ref{fig:ot-setup} when Bob and Eve receive independently erased versions of Alice's transmissions.

The rest of the paper is organized as follows. Section \ref{sec:prob_defn} gives the precise problem definition and states the capacity results that have been proved. Section \ref{sec:ach-protocol} gives the achievability part of the proof of our results, by describing protocols achieving any 2-private and 1-private rate below their respective capacities, for the setp of Figure~\ref{fig:ot-setup}. The converse part for our results are proved in Section \ref{sec:converse}. Most of the rate upper bounds we have hold for the general case of Figure~\ref{fig:ot-setup-bcast}, except for one regime in 1-privacy case where the upper bound is specific to the setup  of Figure~\ref{fig:ot-setup}.

\section{Problem Definition and Statement of Results}
\label{sec:prob_defn}
In the setup of Figure~\ref{fig:ot-setup-bcast}, Alice has two independent, uniformly distributed $m$-length bit strings $K_0,K_1$ and Bob has a uniformly distributed choice bit $C$ independent of $K_0,K_1$. Alice is connected to Bob and Eve by a discrete memoryless broadcast channel defined by the conditional distribution $p_{YZ|X}$. Further, Alice and Bob can communicate over an error-free public channel of unlimited capacity, with Eve able to receive every message sent on this pulic channel. Alice, Bob and Eve are \emph{honest-but-curious} participants in the protocols that run in this setup.
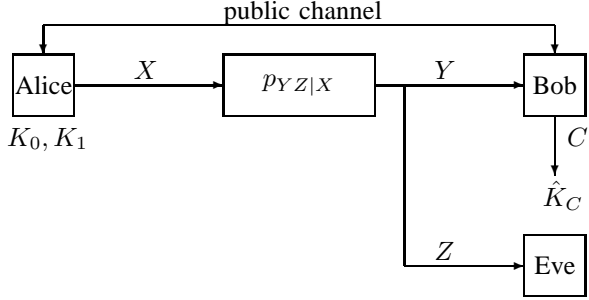
\begin{figure}[h]
\setlength{\unitlength}{0.8cm}
\centering
\begin{picture}(10.5,5.5)
\put(1,3){\framebox(1,1){Alice}}
\put(4.5,3){\framebox(2.5,1){$p_{YZ|X}$}}
\put(9.5,3){\framebox(1,1){Bob}}
\put(9.5,0){\framebox(1,1){Eve}}
\put(2,3.5){\vector(1,0){2.5}}
\put(7,3.5){\vector(1,0){2.5}}
\put(1.5,4.5){\vector(0,-1){0.5}}
\put(10,4.5){\vector(0,-1){0.5}}
\put(1.5,4.5){\line(1,0){8.5}}
\put(7.5,3.5){\line(0,-1){3}}
\put(7.5,0.5){\vector(1,0){2}}
\put(10,3){\vector(0,-1){1}}
\put(0.9,2.5){$K_0,K_1$}
\put(10.2,2.5){$C$}
\put(9.8,1.5){$\hat{K}_C$}
\put(3,3.6){$X$}
\put(8,3.6){$Y$}
\put(8,0.6){$Z$}
\put(4.5,4.6){public channel}
\end{picture}
\caption{Setup for obtaining Oblivious Transfer.}
\label{fig:ot-setup-bcast}
\end{figure}

\begin{definition}
An \emph{($m,n,k$) protocol} uses the broadcast channel at some instances $i_1,i_2,..,i_n \in \{1,\ldots,k\}$ and the public channel at instances $\{1, \ldots, k\}\backslash\{i_1,i_2,..,i_n\}$  and takes the following steps :
 \begin{enumerate}
  \item At the begining of the protocol, Alice and Bob generate private random variables $M,N$ respectively, which are independent of each other and all other system variables available.
  \item $i \notin \{i_1,i_2,..,i_n\}$: $F_i = F_i(K_0,K_1,M,F^{i-1})$ is the public message from Alice, if Alice is the one initiating a public message at time $i$.
  \item $i< i_1$: $F_i = F_i(C,N,F^{i-1})$ is the public message from Bob, if Bob is the one initiating a public message at time $i$.
  \item $i = i_j$: $X_j = X_j(K_0,K_1,M,F^{i-1})$, $F_i = \emptyset$.
  \item $i_j < i < i_{j+1}$: $F_i = F_i(C,N,F^{i-1},Y^j)$ is the public message from Bob, if Bob is the one initiating a public message at time $i$.
 \end{enumerate}
\end{definition}
The protocol computes $\hat{K}_c = \hat{K}(C,N,F^k,Y^n)$ as Bob's string at the end.

We define the \emph{views} of Alice, Bob and Eve at the end of the protocol to be, respectively,
\begin{align*}
  U_k = (K_0,K_1,M,F^k),\,
  V_k = (C,N,F^k,Y^n),\,
  W_k = (F^k,Z^n).
\end{align*}
\begin{definition}
\label{defn:ach2p_rate}
A non-negative number $R_{2P}$ is said to be an \emph{achievable 2-private rate} if there exists a sequence of ($m,n,k$) protocols,  with $\frac{m}{n} \longrightarrow R_{2P}$ as $n \longrightarrow \infty$, such that
\begin{align}
    P[\hat{K}_C \neq K_C] &\longrightarrow 0,
    \label{eqn:2ach_rate_1}\\
    I(K_{\overline{C}} ; V_k,W_k) &\longrightarrow 0,
    \label{eqn:2ach_rate_2}\\
    I(C ; U_k,W_k) &\longrightarrow 0,
    \label{eqn:2ach_rate_3}\\
    I(K_0,K_1,C ; W_k) &\longrightarrow 0,
    \label{eqn:2ach_rate_4}
\end{align}
where $\overline{C}=C\oplus 1$.
\end{definition}

\begin{definition}
The \emph{2-private capacity}, $C_{2P}$ is defined as the supremum of all achievable 2-private rates.
\end{definition}

\begin{definition}
\label{defn:ach1p_rate}
A non-negative number $R_{1P}$ is said to be an \emph{achievable 1-private rate} if there exists a sequence of ($m,n,k$) protocols,  with $\frac{m}{n} \longrightarrow R_{1P}$ as $n \longrightarrow \infty$, such that
 \begin{align}
    P[\hat{K}_C \neq K_C] &\longrightarrow 0,
    \label{eqn:1ach_rate_1}\\
    I(K_{\overline{C}} ; V_k) &\longrightarrow 0,
    \label{eqn:1ach_rate_2}\\
    I(C ; U_k) &\longrightarrow 0,
    \label{eqn:1ach_rate_3}\\
    I(K_0,K_1,C ; W_k) &\longrightarrow 0,
    \label{eqn:1ach_rate_4}
 \end{align}
where $\overline{C}=C\oplus 1$.
\end{definition}

\begin{definition}
The \emph{1-private capacity}, $C_{1P}$ is defined as the supremum of all achievable 1-private rates.
\end{definition}

Our main result is the characterization of $C_{2P}$ and $C_{1P}$ for the setup of Figure~\ref{fig:ot-setup}. In this specific version of the setup of Figure~\ref{fig:ot-setup-bcast}, the broadcast channel is made up of two independent binary erasure channels (BECs). A BEC with erasure probability $\epsilon_1$ (BEC($\epsilon_1$)) connects Alice to Bob and a BEC($\epsilon_2$) connects Alice to Eve. BEC($\epsilon_1$) acts independently of BEC($\epsilon_2$) and no assumption is made on the relative values of $\epsilon_1$ and $\epsilon_2$.

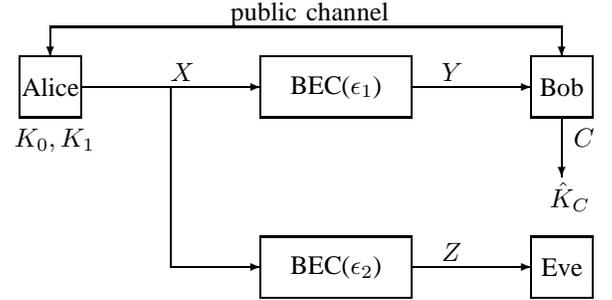
\begin{figure}[h]
\setlength{\unitlength}{0.8cm}
\centering
\begin{picture}(10.5,5.5)
\put(1,3){\framebox(1,1){Alice}}
\put(5,3){\framebox(2.5,1){BEC($\epsilon_1$)}}
\put(9.5,3){\framebox(1,1){Bob}}
\put(5,0){\framebox(2.5,1){BEC($\epsilon_2$)}}
\put(9.5,0){\framebox(1,1){Eve}}
\put(2,3.5){\vector(1,0){3}}
\put(7.5,3.5){\vector(1,0){2}}
\put(1.5,4.5){\vector(0,-1){0.5}}
\put(10,4.5){\vector(0,-1){0.5}}
\put(1.5,4.5){\line(1,0){8.5}}
\put(3.5,3.5){\line(0,-1){3}}
\put(3.5,0.5){\vector(1,0){1.5}}
\put(7.5,0.5){\vector(1,0){2}}
\put(10,3){\vector(0,-1){1}}
\put(0.9,2.5){$K_0,K_1$}
\put(10.2,2.5){$C$}
\put(9.8,1.5){$\hat{K}_C$}
\put(3.5,3.6){$X$}
\put(8,3.6){$Y$}
\put(8,0.6){$Z$}
\put(4.5,4.6){public channel}
\end{picture}
\caption{Setup with the broadcast channel made up of two independent BECs.}
\label{fig:ot-setup}
\end{figure}

We prove the following theorems for the setup of Figure~\ref{fig:ot-setup}.
\begin{theorem}
\label{thm:2p_capacity}
\[ C_{2P} = \epsilon_2 \cdot \min \{\epsilon_1, 1 - \epsilon_1 \}.\]
\end{theorem}

\begin{theorem}
\label{thm:1p_capacity}
\[C_{1P} = \left\{ \begin{array}{ll} \epsilon_1, & \epsilon_1 < \frac{\epsilon_2}{2} \\ \frac{\epsilon_2}{2}, & \frac{\epsilon_2}{2} \leq \epsilon_1 < \frac{1}{2} \\ \epsilon_2(1 - \epsilon_1), & \epsilon_1 \geq \frac{1}{2}  \end{array} \right.\]
\end{theorem}

\section{Proof of Achievability}
\label{sec:ach-protocol}

We begin by briefly reviewing the achievable protocol for a BEC($\epsilon_1$) given in \cite{ot2007}. We note that the setup of Figure~\ref{fig:ot-setup},  for $\epsilon_2 = 1$, reduces to the setup in \cite{ot2007}. Suppose Bob wishes to obtain one of the strings of length $m$ from Alice. Alice transmits an i.i.d. uniform sequence of bits $X^n$ over the broadcast channel. Bob receives an erased version $Y^n$, of $X^n$. Bob will choose two sets of $m$ distinct indices (bit locations) each, a set $B$ (bad set) from the erased indices and a set $G$ (good set) from the unerased indices. Bob chooses these sets uniformly at random from among the possible choices. In this sketch we ignore the possibility that sufficient number of erased and unerased locations are not present; the probability of such an event will be made small enough by an appropriate choice of $m$ and $n$ in the sequel. For $C=0$, Bob assigns $(L_0,L_1) = (G,B)$; otherwise $(L_0,L_1) = (B,G)$. Bob sends $(L_0,L_1)$ to Alice. Alice will form {\em OT keys}, $T_0 = X^n|_{L_0}$ and $T_1 = X^n|_{L_1}$, where $X^n|_{L_0}$ denotes the sequence $X^n$ restricted to the locations in $L_0$. Alice sends $K_0 \oplus T_0$ and $K_1 \oplus T_1$ to Bob over the public channel. Since Bob knows $T_C$, he can obtain $K_C$. It is easy to verify that Alice obtains no information about $C$ and Bob obtains no information about the $K_{\overline{C}}$. 

In the setup of Figure~\ref{fig:ot-setup}, privacy against Eve is additionally required. Let us consider the case of 2-privacy first. In the above scheme, Eve will learn approximately a fraction $1-\epsilon_2$ of both the OT keys $T_0$ and $T_1$. Hence, Alice must additionally protect both the strings before sending them over the public channel. This will be accomplished by setting up additional secret keys (independent of $T_0,T_1$) which are secret from Eve as follows: 
Alice and Bob will create (as explained later) two independent {\em secret keys} $S_0$,$S_1$ (each of length approximately $m(1 - \epsilon_2)$), neither of which is known to Eve and only one of which, namely $S_C$, is known to Bob. Notice that Alice will remain unaware of the identity of the secret key known to Bob. Alice uses these secret keys to further encrypt the strings before sending them over the public channel. Specifically, $S_0$ is used to further encrypt $K_0\oplus T_0$ and $S_1$ for $K_1\oplus T_1$. This is done by   
Alice \emph{expanding} $S_0$, $S_1$ (each of length approximately $m(1 - \epsilon_2)$) to $\tilde{S}_0,\tilde{S}_1$ respectively (each of length $m$ bits), using a binary code (obtained using random coding argument) of rate about ($1-\epsilon_2$). Alice sends $K_0 \oplus T_0 \oplus \tilde{S}_0$ and $K_1 \oplus T_1 \oplus \tilde{S}_1$ to Bob over the public channel.

To generate $S_0,S_1$, Alice and Bob use the secret key agreement scheme of \cite{MaurerWolf2000}, which shows that Alice and Bob can agree upon a secret key which is almost perfectly secret from Eve (i.e. $I(S_i ; W_k) \longrightarrow 0$ as $n \longrightarrow \infty$, $i=1,2$), at rate $\epsilon_2$, if Bob received unerased transmissions from Alice.  So, to generate $S_0,S_1$, Bob will uniformly at random choose sets of indices $G_S$ (good secret key set) and $B_S$ (bad secret key set), of length approximately $\frac{m(1 - \epsilon_2)}{\epsilon_2}$ each, from, respectively, unerased and erased indices which were unused for OT key generation. As before, $G_S,B_S$ are sent to Alice in an order determined by $C$. Alice uses $X^n|_{G_S}$, $X^n|_{B_S}$ to generate the secret keys.

For 1-privacy, Bob may choose $B_S$ randomly from the erased {\em and} unerased  indices it has leftover after creating $G$, $B$ and $G_S$. The rest of the protocol remains the same as that for 2-privacy. Clearly, Bob may now know some or all of the bad secret key. Since this secret key is meant to provide security against Eve with whom Bob does not collude now, the secrecy condition is unaffected.

\subsection{Protocol for a achieving any 2-private rate $r < C_{2P}$}
\label{sec:ach_protocol_2P}

We now present a protocol which achieves any 2-private rate less than $C_{2P}$, for the setup of Figure~\ref{fig:ot-setup}. 
For any $\delta \in (0,1)$, we define
\begin{equation*}
\tilde{\epsilon}_2  := \epsilon_2(1 - \delta). 
\end{equation*}
In a protocol where Alice transmits $n$ bits over the BEC, let  $E, \overline{E}, E'$ be, respectively, the set of indices where Bob sees erasures, Bob sees non-erasures, and Eve sees erasures.
\begin{align*}
E  &:= \{ i \in \{i_1, \ldots, i_n\} : Y_i = \text{erasure} \},\\
\overline{E}  &:= \{ i \in \{i_1, \ldots, i_n\} : Y_i \neq \text{erasure} \},\\
E' &:= \{ i \in \{i_1, \ldots, i_n\} : Z_i = \text{erasure} \}.
\end{align*}
Let $\mathbb{U}(A)$ denotes a uniformly random choice from the set $A$.

The following lemma says that, with high probability, Eve will see at least $\tilde{\epsilon}_2$ fraction of its received sequence erased.
\begin{lemma}
\label{lem:provision_2}
\begin{equation*}
P\left[ \frac{|E'|}{n} \geq \tilde{\epsilon}_2 \right] \longrightarrow 1 \text{ exponentially in $n$}.
\end{equation*}
\end{lemma}
\begin{proof}
The claim follows from Chernoff bound.
\end{proof}

The following lemma says that for any rate $r < C_{2P}$ and a suitably low $\delta$ (to define $\tilde{\epsilon}_2$), Bob will have enough erased and unerased $Y_i$'s with which to run the protocol and achieve rate $r$. 
\begin{lemma}
\label{lem:provision_3}
Suppose $r < C_{2P}$ and $\delta < (1 - \frac{r}{C_{2P}})$. Then
\begin{align*}
 P \left[ |E| \geq \frac{nr}{\tilde{\epsilon_2}} \right] & \longrightarrow 1 \text{ exponentially in $n$}, \\
 P \left[ |\overline{E}| \geq \frac{nr}{\tilde{\epsilon_2}} \right] & \longrightarrow 1 \text{ exponentially in $n$}.
\end{align*}
\end{lemma}

\begin{proof}
The claims follow from Chernoff bound.
\end{proof}

\begin{protocol}(Protocol for achieving any 2-private rate $r < C_{2P}$, for the setup in Figure~\ref{fig:ot-setup})
\label{sch:our_2P}

\emph{Protocol parameters} (known to all parties): rate $r$, $\delta$ (suitably low, as per Lemma~\ref{lem:provision_3}), $\tilde{\epsilon}_2$, a binary ($nr, nr(1 - \tilde{\epsilon}_2)$)-code $\Lambda_{nr}$ chosen via a random coding argument

\begin{description}
\item[\textbf{Alice}] Transmits an i.i.d. sequence $X^n$, where $\forall i, \;\; X_i \thicksim \mathbb{U}(\{0,1\})$, over the BEC.

\item[\textbf{Bob}] Receives the $Y^n$ from BEC($\epsilon_1$). Let $r < C_{2P}$. Bob now creates the following sets:
              \begin{IEEEeqnarray*}{rCl}
               G & \thicksim & \mathbb{U} \left ( \left \{A \subset \overline{E}: |A| = nr \right \} \right ), \\
               G_S & \thicksim & \mathbb{U} \left ( \left \{A \subset \overline{E} \backslash G: |A| = \frac{nr(1 - \tilde{\epsilon_2})}{\tilde{\epsilon_2}} \right \}  \right ), \\
               B  & \thicksim & \mathbb{U} \left (\left \{ A \subset E: |A| = nr \right \} \right ), \\
               B_S & \thicksim &  \mathbb{U} \left ( \left \{ A \subset E \backslash B : |A| = |G_S| \right \} \right ).
              \end{IEEEeqnarray*}
      Bob has sufficiently many erased and unerased $Y_i$'s (with high probability) to create these sets, as a consequence of Lemma~\ref{lem:provision_3}.
     Then, depending on the value of $C$, Bob further creates the sets $L_{00}, L_{01}, L_{10}, L_{11}$ as follows.
	\begin{align*}
	C = 0:\qquad &L_{00} = G,\quad L_{01} = G_S \\
                     &L_{10} = B,\quad L_{11} = B_S \\  
        C = 1:\qquad &L_{00} = B,\quad L_{01} = B_S \\
                     &L_{10} = G,\quad L_{11} = G_S 
        \end{align*}
     Bob sends $L_{00},L_{01}, L_{10},L_{11}$ to Alice over the public channel.
\item[\textbf{Alice}] computes the following keys
      \begin{align*}
        T_0 & = X^n |_{L_{00}}, \\
        T_1 & = X^n |_{L_{10}}. 
      \end{align*}
       Alice generates secret key $S_0$ from $X^n|_{L_{01}}$, assuming Bob knows $X^n|_{L_{01}}$.
       Alice also generates secret key $S_1$ from $X^n|_{L_{11}}$, assuming Bob knows $X^n|_{L_{11}}$. $S_0,S_1$ are $nr(1-\tilde{\epsilon}_2)$ bits each.

       Alice  expands the secret keys $S_0$,$S_1$, to get $\tilde{S}_0, \tilde{S}_1$ of $nr$ bits each
      \begin{align*}
       \tilde{S}_0 = \Lambda_{nr}(S_0), \\
       \tilde{S}_1 = \Lambda_{nr}(S_1).
      \end{align*}
      Alice finally sends the following two strings to Bob over the public channel:
      \begin{align*}
        K_0 & \oplus T_0 \oplus \tilde{S}_0, \\ 
        K_1 & \oplus T_1 \oplus \tilde{S}_1.
      \end{align*}

\item[\textbf{Bob}] has the pair ($T_C, \tilde{S}_C$), thus it can get $K_C$.
\end{description}
\end{protocol}

\begin{lemma}
\label{lem:any_r_2p_ach}
Any $r < C_{2P}$ is an achievable 2-private rate..
\end{lemma}

\begin{proof}

A sequence of instances $\{P_n\}_{n \in \mathbb{N}}$ of Protocol~\ref{sch:our_2P} will be used. Let $J := \mathbbm{1}_{\{|E| \geq nr/\tilde{\epsilon_2}\}\cap \{|\overline{E}| \geq nr/\tilde{\epsilon_2}\}}$ be the random variable to indicate the event that Bob has seen enough erasures and non-erasures. By Lemma~\ref{lem:provision_3}, $Pr [J=1] \longrightarrow 1$ exponentially fast. 

For ease of notation, we denote:
\begin{align*}
\underline{G} & = (G,B,G_S,B_S) \\
\underline{L} & = (L_{00}, L_{01}, L_{10}, L_{11}) \\
\tilde{K}_C & =  K_C \oplus T_C \oplus \tilde{S}_C \\
\tilde{K}_{\overline{C}} & =  K_{\overline{C}} \oplus T_{\overline{C}} \oplus \tilde{S}_{\overline{C}}
\end{align*}

We note that for our achievable scheme,
\begin{align*}
U_k & = (K_0,K_1,X^n, \underline{L}, \tilde{K}_0, \tilde{K}_1) \\
V_k & = (C,Y^n,\underline{G}, \tilde{K}_0, \tilde{K}_1) \\ 
W_k & = (Z^n,\underline{L}, \tilde{K}_0, \tilde{K}_1)
\end{align*}

\begin{enumerate}

\item To show that (\ref{eqn:2ach_rate_1}) is satisfied for $\{P_n\}_{n \in \mathbb{N}}$, we note that
\begin{align*}
P[\hat{K}_C \neq K_C] & = P[J=0]P[\hat{K}_C \neq K_C | J = 0] \\ & + P[J=1]P[\hat{K}_C \neq K_C | J = 1]
\end{align*} 

Since $Pr[J=0] \rightarrow 0$ exponentially fast, it is sufficient to show that $P[\hat{K}_C \neq K_C | J = 1] \longrightarrow 0$ as $n \longrightarrow \infty$.

Now, when Bob does have sufficient erasures and non-erasures, Bob knows $T_C$ since $T_C = X^n|_G$. Similarly, Bob knows $\tilde{S}_C$ since $\tilde{S}_C$ is a function of $S_C$ which, in turn, is a function of $X^n|_{G_S}$. Thus, when Bob receives $\tilde{K}_C$ from Alice, Bob learns $K_C$ with zero error. Hence, $P[\hat{K}_C \neq K_C | J = 1] = 0$.

\item  \label{item:proof_1} To show that (\ref{eqn:2ach_rate_2}) is satisfies for $\{P_n\}_{n \in \mathbb{N}}$, we note that
\begin{align*}
I(K_{\overline{C}} ; V_k,W_k ) & \leq I(K_{\overline{C}} ; V_k,W_k, J) \\
& = \sum_{j=0,1}Pr[J=j] \,I(K_{\overline{C}} ; V_k,W_k| J=j) \\
& \hspace*{15mm}+ I(K_{\overline{C}} ; J) .
\end{align*}
Since $Pr[J=0] \rightarrow 0$ exponentially fast and $I(K_{\overline{C}} ; J)=0$, it is sufficient
to show that $I(K_{\overline{C}}  ;  V_k,W_k | J = 1) \longrightarrow 0$ as $n \longrightarrow \infty$. Now,

\begin{IEEEeqnarray*}{rCl}
I&(&K_{\overline{C}}  ;  V_k,W_k | J = 1) \\
&& = I(K_{\overline{C}} ; C,Y^n, Z^n, \underline{G}, \underline{L}, \tilde{K}_0, \tilde{K}_1 | J = 1) \\
&& = I(K_{\overline{C}} ; C,Y^n, Z^n, \underline{G}, \tilde{K}_C, \tilde{K}_{\overline{C}} | J = 1) \\
&& = I(K_{\overline{C}} ; \tilde{K}_{\overline{C}} | C,Y^n, Z^n, \underline{G}, \tilde{K}_C , J = 1) \\
&& \text{  [since $K_{\overline{C}}$ is indep. of ($C,Y^n, Z^n, \underline{G}, \tilde{K}_C, J$)]} \\
&& = H(\tilde{K}_{\overline{C}} | C,Y^n, Z^n, \underline{G}, \tilde{K}_C, J = 1) \\ && - H(\tilde{K}_{\overline{C}} | K_{\overline{C}}, C,Y^n, Z^n, \underline{G}, \tilde{K}_C, J = 1) \\
&& = nr - H(\tilde{K}_{\overline{C}} | K_{\overline{C}}, C,Y^n, Z^n, \underline{G}, \tilde{K}_C, J = 1) \\
&& \text{  [since $K_{\overline{C}}$ is indep. of ($C,Y^n, Z^n, \underline{G}, \tilde{K}_C, J$) and } \\ && \text{ is uniform over its alphabet ]} \\
&& = nr - H(T_{\overline{C}} \oplus \tilde{S}_{\overline{C}} | K_{\overline{C}}, C,Y^n, Z^n, \underline{G}, \tilde{K}_C, J = 1) \\
&& = nr - H(T_{\overline{C}} \oplus \tilde{S}_{\overline{C}} | C, Y^n, Z^n, \underline{G}, \tilde{K}_C, J = 1) \\
&& \text{  [since $K_{\overline{C}}$ is indep. of ($T_{\overline{C}}, \tilde{S}_{\overline{C}}, C, Y^n, Z^n, \underline{G}, \tilde{K}_C, J$) ]} \\
&& = nr - H(T_{\overline{C}} \oplus \tilde{S}_{\overline{C}} | C,Z^n, \underline{G}, J = 1) \\
&& \text{  [as $T_{\overline{C}}, \tilde{S}_{\overline{C}} - C,Z^n,\underline{G} - T_C, \tilde{S}_C, K_C, Y^n$ is a} \\&& \text{ Markov chain conditioned on  J = 1]} \\
&& = nr - H(T_{\overline{C}} \oplus \tilde{S}_{\overline{C}} | C,Z^n, \underline{G}, I, J = 1) \\
&& \text{  [$I :=  \{i \in \{1,2,\ldots,nr\} : T_{\overline{C},i}$ seen unerased in $Z^n\}$,} \\ && \text{  $I$ is a function of ($Z^n,B$)]} \\
&& = nr - H(\tilde{S}_{\overline{C},I} | C,Z^n, \underline{G}, I, J = 1) \\ && - H(T_{\overline{C},\overline{I}} \oplus \tilde{S}_{\overline{C}, \overline{I}} | C,Z^n, \underline{G}, I, \tilde{S}_{\overline{C},I}, J = 1)\\
&& \text{  [$\overline{I} = \{1,2,\ldots,nr\} \backslash I$]} \\
&& = nr - H(\tilde{S}_{\overline{C},I} | C,Z^n, \underline{G}, I, J = 1) - nr\epsilon_2 \\
&& \text{  [since $T_{\overline{C},\overline{I}} - I,B - C,Z^n, \underline{G} \backslash B, \tilde{S}_{\overline{C},I}$ is a Markov chain} \\ && \text{ conditioned on J = 1,  and $T_{\overline{C},\overline{I}}$ is uniform over its} \\ && \text{  alphabet]} \\
&& = nr(1 - \epsilon_2) - H(\tilde{S}_{\overline{C},I} | C,I, J = 1) \\ && + I(\tilde{S}_{\overline{C},I} ; Z^n, \underline{G} | C,I, J = 1) \\
&& = nr(1 - \epsilon_2) - H(\tilde{S}_{\overline{C},I} | I, J = 1) \\ && + I(\tilde{S}_{\overline{C},I} ; Z^n, \underline{G} | C,I, J = 1) \\
&& \text{  [since $X^n|_{B_S} - I - C$ is a Markov chain conditioned on} \\ && \text{ $J = 1$ and $\tilde{S}_{\overline{C}}$  is a function of $X^n|_{B_S}$] } \\
&& \leq nr(1 - \epsilon_2) - \sum_{|i| < nr(1-\tilde{\epsilon}_2 - \frac{\epsilon_2 \delta}{2})} p_{I}(i) H(\tilde{S}_{\overline{C},I} | I = i, J = 1)  \\ && + I(\tilde{S}_{\overline{C},I} ; Z^n, \underline{G} | C,I, J = 1) \\
&& \leq nr(1 - \epsilon_2) - \sum_{|i| < nr(1-\tilde{\epsilon}_2 - \frac{\epsilon_2 \delta}{2})} p_{I}(i) \cdot (|i| - 2^{-nr\frac{\epsilon_2 \delta}{4}}) \\ && + I(\tilde{S}_{\overline{C},I} ; Z^n, \underline{G} | C,I, J = 1) \\
&& \text{[using Lemma~\ref{lem:random_coding_lemma}, since $\frac{\epsilon_2 \delta}{4} < 1-\tilde{\epsilon}_2 - \frac{|i|}{nr}$]} \\
&& = nr(1 - \epsilon_2) + \sum_{|i| < nr(1-\tilde{\epsilon}_2 - \frac{\epsilon_2 \delta}{2})} p_{I}(i) 2^{-nr\frac{\epsilon_2 \delta}{4}} - \mathbb{E}(|I|) \\ && + \sum_{|i| \geq nr(1-\tilde{\epsilon}_2 - \frac{\epsilon_2 \delta}{2})}  p_{I}(i) |i| + I(\tilde{S}_{\overline{C},I} ; Z^n, \underline{G} | C,I, J = 1) \\
&& \leq  nr(1 - \epsilon_2) +  \overline{p}_n 2^{-nr\frac{\epsilon_2 \delta}{4}} - nr(1 - \epsilon_2) \\ && + \sum_{|i| \geq nr(1-\tilde{\epsilon}_2 - \frac{\epsilon_2 \delta}{2})} p_{I}(i) \cdot nr + I(\tilde{S}_{\overline{C},I} ; Z^n, \underline{G} | C,I, J = 1) \\
&& \text{  [$p_n =  Pr(|I| \geq nr(1-\tilde{\epsilon}_2 - \frac{\epsilon_2 \delta}{2}))$,  $\overline{p}_n = 1 - p_n$,} \\ && \text{ $p_n \longrightarrow 0$ exp. as $n \longrightarrow \infty$ by Lemma~\ref{lem:provision_2}]} \\
&& = \overline{p}_n 2^{-nr\frac{\epsilon_2 \delta}{4}} + p_n \cdot nr + I(\tilde{S}_{\overline{C},I} ; Z^n, \underline{G} | C,I, J = 1) \\
&& \leq \overline{p}_n 2^{-nr\frac{\epsilon_2 \delta}{4}}  + p_n \cdot nr  +  I(\tilde{S}_{\overline{C}} ; Z^n, \underline{G} | C,I, J = 1) \\
&& = \overline{p}_n 2^{-nr\frac{\epsilon_2 \delta}{4}} + p_n \cdot nr  +  I(\tilde{S}_{\overline{C}} ; Z^n, \underline{G},I | C, J = 1) \\
&& \text{  [since $I$ is indep. of ($\tilde{S}_{\overline{C}}, C$), conditioned on $J=1$]} \\
&& = \overline{p}_n 2^{-nr\frac{\epsilon_2 \delta}{4}}  + p_n \cdot nr  +  I(\tilde{S}_{\overline{C}} ; Z^n, \underline{G} | C, J = 1) \\
&& \text{  [since $I$ is a function of ($Z^n, B$)]} \\
&& \leq \overline{p}_n 2^{-nr\frac{\epsilon_2 \delta}{4}}  +  p_n \cdot nr  + I(\tilde{S}_{\overline{C}} ; Z^n, \underline{G}, C | J = 1) \\
&& = \overline{p}_n 2^{-nr\frac{\epsilon_2 \delta}{4}}  + p_n \cdot nr  +  I(\tilde{S}_{\overline{C}} ; Z^n, B_S | J = 1) \\
&& \text{  [since $\tilde{S}_{\overline{C}} - Z^n,B_S - C,G,B,G_S$ is a Markov chain,} \\ && \text{ conditioned on $J=1$]} \\
&& \leq \overline{p}_n 2^{-nr\frac{\epsilon_2 \delta}{4}}  +  p_n \cdot nr  +  I(\tilde{S}_{\overline{C}}, S_{\overline{C}} ; Z^n, B_S | J = 1) \\
&& = \overline{p}_n 2^{-nr\frac{\epsilon_2 \delta}{4}}  +  p_n \cdot nr  +  I(S_{\overline{C}} ; Z^n, B_S | J = 1) \\
&& \text{  [because $\tilde{S}_{\overline{C}}$ is a function of $S_{\overline{C}}$]}
\end{IEEEeqnarray*}

The R.H.S. goes to zero, as $n \longrightarrow \infty$, since $p_n \longrightarrow 0$ exponentially fast with $n$ and since the information which Eve receives about $S_{\overline{C}}$ can be made arbitrarily small for sufficiently large $n$(\cite{MaurerWolf2000}).

\item \label{item:proof_2} To show that eqn. (\ref{eqn:2ach_rate_3}) is satisfied for $\{P_n\}_{n \in \mathbb{N}}$, we note that
\begin{align*}
I(C ; U_k,W_k ) & =  I(C ; U_k) \\ & \text{  [since $C - U_k - W_k$ is a Markov chain]}\\ 
& \leq I(C ; U_k,J) \\
& = \sum_{j=0,1}Pr[J=j] \,I(C ; U_k| J=j) \\
& \hspace*{15mm}+ I(C ; J) .
\end{align*}
Since $Pr[J=0] \rightarrow 0$ exponentially fast, $I(C ; U_k| J=0) \leq 1$ and $I(C ; J)=0$, it is sufficient
to show that $I(C  ;  U_k | J = 1) \longrightarrow 0$ as $n \longrightarrow \infty$

Now,

\begin{IEEEeqnarray*}{rCl}
I(C &;& U_k| J = 1) \\ 
               & = & I(C ; K_0,K_1, X^n, \underline{L}, \tilde{K}_0, \tilde{K}_1 | J = 1) \\
               & = & I(C ; \underline{L}, \tilde{K}_0, \tilde{K}_1 | K_0,K_1, X^n, J = 1) \\
               && \text{ [since $C$ is indep. of ($K_0,K_1,X^n, J$)]} \\
               & = & I(C ; \underline{L}| K_0,K_1, X^n, J = 1) \\
               && \text{ [since $T_0,T_1,\tilde{S}_0, \tilde{S}_1$ are functions of ($X^n,\underline{L}$)]} \\
               & = & I(C ; \underline{L}| X^n, J = 1) \\
               && \text{ [since ($K_0,K_1$) is indep. of ($C,X^n,\underline{L}, J$)]} \\
               & = & H(\underline{L}| X^n, J = 1) - H(\underline{L} | X^n,C, J = 1) \\
               & = & H(\underline{L} | J = 1) - H(\underline{G} | C, J = 1) \\
               && \text{ [since $X^n$ is indep. of ($\underline{L}, J, C$)]} \\
               & = & H(\underline{L} | J = 1) - H(\underline{G} | J = 1) \\
               && \text{ [since $C$ is indep. of ($\underline{G}, J$)]} \\
               & = & 0
\end{IEEEeqnarray*}

\item To show that eqn. (\ref{eqn:2ach_rate_4}) is satisfied for $\{P_n\}_{n \in \mathbb{N}}$, we will use the proofs already seen for Part (\ref{item:proof_1}) and Part (\ref{item:proof_2}), as follows:

\begin{IEEEeqnarray*}{rCl}
I&(& K_0,K_1,C ; W_k) \\ 
&& = I(K_0,K_1 ; W_k | C) + I(C ; W_k) \\
&& = I(K_C, K_{\overline{C}} ; Z^n, \underline{L}, \tilde{K}_C, \tilde{K}_{\overline{C}} | C) + I(C ; W_k) \\
&& = I(K_C, K_{\overline{C}} ;  \tilde{K}_C, \tilde{K}_{\overline{C}} | C, Z^n, \underline{L})+ I(C ; W_k) \\
&& \text{ [since ($K_C, K_{\overline{C}}$) is indep. of ($Z^n, \underline{L}, C$)]} \\
&& = I(K_C ; \tilde{K}_C, \tilde{K}_{\overline{C}} | C, Z^n, \underline{L}) \\ && + I(K_{\overline{C}} ; \tilde{K}_C, \tilde{K}_{\overline{C}} | C, Z^n, \underline{L}, K_C) +   I(C ; W_k) \\
&& = I(K_C ; \tilde{K}_C | C, Z^n, \underline{L}) + I(K_C ; \tilde{K}_{\overline{C}} | C, Z^n, \underline{L}, \tilde{K}_C) \\ && + I(K_{\overline{C}} ; \tilde{K}_{\overline{C}} | C, Z^n, \underline{L}, K_C) \\ && + I(K_{\overline{C}} ; \tilde{K}_C | C, Z^n, \underline{L}, K_C, \tilde{K}_{\overline{C}})  + I(C ; W_k) \\
&& = I(K_C ; \tilde{K}_C | C, Z^n, \underline{L})  + I(K_{\overline{C}} ; \tilde{K}_{\overline{C}} | C, Z^n, \underline{L}, K_C) \\ && + I(K_{\overline{C}} ; T_C \oplus \tilde{S}_C | C, Z^n, \underline{L}, K_C, \tilde{K}_{\overline{C}}) + I(C ; W_k) \\
&& \text{ [since $K_{\overline{C}}$ is indep. of ($C,Z^n,\underline{L}, \tilde{K}_C, K_C$)]} \\
&& = I(K_C ; \tilde{K}_C | C, Z^n, \underline{L})  + I(K_{\overline{C}} ; \tilde{K}_{\overline{C}} | C, Z^n, \underline{L}, K_C) \\ && + I(C ; W_k) \\
&& \text{ [since $T_C,\tilde{S}_C - C,Z^n,\underline{L} - K_C, K_{\overline{C}}, T_{\overline{C}}, \tilde{S}_{\overline{C}}$ } \\ && \text{ is  a Markov chain]} \\
&& = I(K_C ; \tilde{K}_C | C, Z^n, \underline{L})  + I(K_{\overline{C}} ; \tilde{K}_{\overline{C}} | C, Z^n, \underline{L}) \\ && + I(C ; W_k) \\
&& \text{ [since $K_C$ is indep. of ($C, Z^n, \underline{L}, K_{\overline{C}}, \tilde{K}_{\overline{C}}$)]} \\
\end{IEEEeqnarray*}

The first two terms above go to zero following the exact same arguments as in the proof of Part (\ref{item:proof_1}). The last term above goes to zero as a consequence of the proof of Part (\ref{item:proof_2}).

\end{enumerate}

\end{proof}

\subsection{Protocol for achieving any 1-private rate $r < C_{1P}$}
\label{sec:ach_protocol_1P}
The difference in this protocol, compared to Protocol \ref{sch:our_2P}, is in the way the set $B_S$ is chosen.
\begin{itemize}
 \item $\epsilon_1 < \frac{\epsilon_2}{2}$: Bob chooses $B_S$ randomly out of leftover unerased indices and, thus, fully knows the corresponding secret key.
 \item $\frac{\epsilon_2}{2} \leq \epsilon_1 < \frac{1}{2}$: Bob chooses $B_S$ randomly out of all indices left after creating $G$,$B$ and $G_S$ and may know the corresponding secret key partially.
 \item $\epsilon_1 \geq \frac{1}{2}$: Bob chooses $B_S$ randomly out of leftover erased indices and knows nothing about the corresponding secret key.
\end{itemize}
Just as in Lemma~\ref{lem:provision_3}, we can show that for a suitably low $\delta$ and sufficiently large $n$, Bob will have enough erased and unerased $Y_i$'s with which to run the protocol and achieve any rate $r < C_{1P}$. 
\begin{lemma}
\label{lem:any_r_1p_ach}
Any $r < C_{1P}$ is an achievable 1-private rate.
\end{lemma}
The proof of this lemma is similar to the proof of Lemma~\ref{lem:any_r_2p_ach}.

\section{Proof of Converse}
\label{sec:converse}
The converses hold even under a weaker sense of security where conditions
\eqref{eqn:2ach_rate_2}, \eqref{eqn:2ach_rate_4}, \eqref{eqn:1ach_rate_2}, and
\eqref{eqn:1ach_rate_4} hold only with a $\frac{1}{n}$ factor on the
left-hand-side.

\begin{lemma}
For the setup of Figure~\ref{fig:ot-setup},
\begin{equation*}
 C_{2P} \leq \epsilon_2 \cdot \min \{\epsilon_1, 1 - \epsilon_1 \}
\end{equation*}
\end{lemma}
\begin{proof}
We first state a general upperbound on $C_{2P}$:
For the setup of Figure~\ref{fig:ot-setup-bcast},
\begin{align}
 C_{2P} \leq \min\left \{ \max_{p_X} I(X ; Y | Z), \max_{p_X} H(X | Y, Z) \right \}. \label{eq:c2p-general-outer}
\end{align}
$C_{2P} \leq \max_{p_X} I(X ; Y | Z)$ follows from the observation that operating the protocol with Bob setting $C=0$ allows Alice and Bob to agree on the secret key $K_0$ which is secret from Eve. Since $\max_{p_X} I(X ; Y | Z)$ is an upperbound on secret key capacity of the broadcast channel $p_{Y,Z|X}$ (with public discussion)~\cite{sec-key1993}, the bound follows.

It is easy to verify that the 2-private protocol can be viewed as a (two-party) OT protocol between the parties Alice and Bob-Eve (combined). Hence, by invoking an outerbound on OT capacity in~\cite{ot2007}, we have  $C_{2P} \leq \max_{p_X} H(X | Y, Z)$.
Evaluating these upper bounds in \eqref{eq:c2p-general-outer} for the specific setup in Figure~\ref{fig:ot-setup},
\begin{align*}
C_{2P} & \leq \max_{p_X} I(X ; Y | Z)  = \epsilon_2  (1 - \epsilon_1),\\
C_{2P} & \leq \max_{p_X} H(X | (Y, Z)) =  \epsilon_2 \epsilon_1.
\end{align*}
\end{proof}

\begin{lemma}
For the setup of Figure~\ref{fig:ot-setup},
\begin{equation*}
C_{1P} \leq   \min \left\{ \epsilon_1, \quad \frac{\epsilon_2}{2}, \quad \epsilon_2 (1 - \epsilon_1)  \right\} 
\end{equation*}
\end{lemma}
\begin{proof}
We first show that $C_{1P} \leq   \min \left\{ \epsilon_1, \epsilon_2  (1 - \epsilon_1)  \right\}$ by means of the following more general statement:
For the setup of Figure~\ref{fig:ot-setup-bcast},
\begin{align}
 C_{1P} \leq \left \{ \max_{p_X} I(X ; Y | Z), \max_{p_X} H(X | Y) \right \}. \label{eq:c1p-general-outer}
\end{align}
Proof of $ C_{1P} \leq \max_{p_X} I(X ; Y | Z)$ is identical to the one for 2-private case \eqref{eq:c2p-general-outer}. $C_{1P} \leq \max_{p_X} H(X|Y)$ follows from observing that a 1-private protocol is also a protocol for OT between Alice and Bob over the channel $p_{Y|X}$ for which $\max_{p_X} H(X|Y)$ is an upperbound on OT capacity~\cite{ot2007}.
Evaluating \eqref{eq:c1p-general-outer} for the specific setup in Figure~\ref{fig:ot-setup},
\begin{align*}
C_{1P} & \leq \max_{p_X} H(X | Y)  = \epsilon_1, \\
C_{1P} & \leq \max_{p_X} I(X ; Y | Z) = \epsilon_2(1 - \epsilon_1).
\end{align*}
It only remains to show that $C_{1P} \leq \frac{\epsilon_2}{2}$. For this, we need the following lemma which states that ($X^n,F^k$) must together carry nearly all the information about ($K_0,K_1$).
\begin{lemma} \label{lem:small_quant}
$\frac{1}{n}H(K_0,K_1 | X^n,F^k) \longrightarrow 0 \text{ as } n \longrightarrow \infty$.
\end{lemma}
The proof is deferred to the appendix. The lemma can be interpreted as follows:
Bob's privacy against Alice (eqn.~\ref{eqn:1ach_rate_3}) implies that Alice is unaware of which string is required by Bob. This forces that both the strings be decodable from observing the the interface of Alice to the system (i.e. observing ($X^n,F^k$)). If this were not the case and $K_0$, say, could not be fully decoded from ($X^n,F^k$), then Alice can infer that Bob wanted $K_1$ (i.e., $C=1$) violating the requirement of (\ref{eqn:1ach_rate_3}).

To convert this into an upperbound on the rate, intuitively, Eve has access to all of $(X^n,F^k)$ except those bits of $X^n$ erased by her channel. Since the erased fraction of bits is about $\epsilon_2$, and we require both strings to be secret from Eve, each string has a rate of at most $\epsilon_2/2$. We make this argument more formally below.
Let $I=\{i_1,\ldots,i_n\}$, be the instances where the channel is used and the (random) set of indices at which Eve saw erasures be $E'  := \{ i \in I : Z_i = \text{erasure} \}$. Let $e'$ denote a realization of $E'$ and $\overline{e}'=I\backslash e'$ its complement.
\begin{align*}
2m & =  H(K_0,K_1) \\
   & =  I(K_0,K_1 ; X^n,F^k) + H(K_0,K_1 | X^n,F^k) \\
   & \stackrel{\text{(a)}}{=}  I(K_0,K_1 ; X^n,F^k) + o(n) \\ 
   & \stackrel{\text{(b)}}{=}  I(K_0,K_1 ; X^n,F^k | E') + o(n) \\
   & =  \sum_{e' \subseteq I } p_{E'}(e') I(K_0,K_1 ; X^n,F^k | E' = e') + o(n) \\
   & =  \sum_{e' \subseteq I } p_{E'}(e') I(K_0,K_1 ; X^n|_{\overline{e}'}, F^k| E' = e') + o(n)  \\ &\quad +\sum_{e' \subseteq I } p_{E'}(e') I(K_0,K_1 ; X^n|_{e'}| X^n|_{\overline{e}'}, F^k, E' = e')   \\
   & \leq   \sum_{e' \subseteq I } p_{E'}(e') I(K_0,K_1 ; X^n|_{\overline{e}'}, F^k| E' = e') \\ &\quad + \sum_{e' \subseteq I } p_{E'}(e') H(X^n|_{e'}| E' = e') + o(n) \\
   & =    I(K_0,K_1 ; Z^n, F^k) + n\epsilon_2 + o(n) \\
  & \stackrel{\text{(c)}}{=}   n\epsilon_2 + o(n),
\end{align*}
where (a) follows from Lemma~\ref{lem:small_quant}, (b) from the independence of Eve's channel, and (c) from \eqref{eqn:1ach_rate_4}. Therefore, $C_{1P} \leq \frac{\epsilon_2}{2}$.
\end{proof}

\section*{Acknowledgment}
The work  was supported in part by the Bharti Centre for Communication, IIT Bombay, a grant from the Information Technology Research Academy, Media Lab Asia, to IIT Bombay and TIFR, a grant from the Department of Science and Technology, Government of India, to IIT Bombay, and a Ramanujan Fellowship from the Department of Science and Technology, Government of India, to V. Prabhakaran. S. Diggavi was supported in part by NSF awards 1136174, 1321120 and MURI award AFOSR FA9550-09-064.

\appendices

\section{Proof of Lemma~\ref{lem:small_quant}}
We need two lemmas from \cite{ot2007}, which are stated here for completeness.
\begin{lemma}
\label{lem:lemma3_ahl_csis}
Let U,V,Z denote random variables with values in finite sets $\mathcal{U}$, $\mathcal{V}$ and $\mathcal{Z}$ respectively. Suppose $z_1,z_2 \in \mathcal{Z}$ with $P[Z = z_1] = p > 0$ and $P[Z = z_2] = q > 0$. Then,
\begin{IEEEeqnarray*}{rCl}
 | H(U|V,Z = z_1) & - & H(U|V,Z = z_2) | \\ &&  \leq 3 \sqrt{\frac{(p+q)\ln 2}{2pq}I(UV ; Z)} log|\mathcal{U}| + 1.
\end{IEEEeqnarray*}
\end{lemma}
\begin{proof}
See \cite{ot2007}.
\end{proof}

\begin{lemma}
\label{lem:cond_indep}
$I(K_0, K_1 , M ;  C ,N,Y^n | X^n,F^k) = 0.$
\end{lemma}
\begin{proof}
See \cite{ot2007} or Lemma 2.2 of \cite{sec-key1993}.
\end{proof}
Note that (\ref{eqn:1ach_rate_3}) and Lemma~\ref{lem:lemma3_ahl_csis} together imply
\begin{align*}
H(K_0 | X^n,F^k,C = 0) - H(K_0 | X^n,F^k,C = 1) & = o(n), \\
H(K_1 | X^n,F^k,C = 0) - H(K_1 | X^n,F^k,C = 1) & = o(n).
\end{align*}
Multiplying both equations by $\frac{1}{2}$ and subtracting, we get
\begin{equation}
\label{eqn:lemma3_infr3}
H(K_C | X^n,F^k,C) - H(K_{\overline{C}} | X^n,F^k,C) = o(n).
\end{equation}
Lemma~\ref{lem:cond_indep} implies that $I(K_0,K_1 ; C | X^n,F^k) = 0$. Hence,
\begin{align*}
H(K_0,K_1 &| X^n,F^k)\\
                     &=  H(K_0,K_1 | X^n,F^k,C) \\
                     &=  H(K_C,K_{\overline{C}} | X^n,F^k,C) \\
                     &=  H(K_C | X^n,F^k,C) + H(K_{\overline{C}} | X^n,F^k,C, K_C) \\
                     &\leq  H(K_C | X^n,F^k,C) + H(K_{\overline{C}} | X^n,F^k,C).
\end{align*}
In light of (\ref{eqn:lemma3_infr3}), this lemma will be proved if we show either $H(K_C | X^n,F^k,C)$ or $H(K_{\overline{C}} | X^n,F^k,C)$ to be $o(n)$.

For this we note that Lemma~\ref{lem:cond_indep} implies \[I(K_0,K_1 ; N,Y^n | X^n,F^k,C) = 0.\] This, in turn, implies that \[I(K_C,K_{\overline{C}} ; N,Y^n | X^n,F^k,C) = 0.\] Hence, $I(K_C ; N,Y^n | X^n,F^k,C) = 0$. Therefore,
\begin{align*}
H(K_C | X^n,F^k,C) &= H(K_C | X^n,F^k,C, N, Y^n) \\
                   &\stackrel{\text{(a)}}{=}  H(K_C | X^n,F^k,C, N, Y^n, \hat{K}_C) \\
                   &\leq  H(K_C | \hat{K}_C) \\
                   &\stackrel{\text{(b)}}{=}  o(n), \end{align*}
where (a) follows from the fact that since $\hat{K}_C$ is a function of ($C,N,Y^n,F^k$), and (b) from \eqref{eqn:1ach_rate_1} and Fano's inequality.

\section{A useful Lemma}
\label{sec:proof_of_random_coding_lemma}
In the following, for $r' < 1$, we consider the construction of a 
random ($n, nr'$) binary code ${\cal C}$, generated with i.i.d. Ber($\frac{1}{2}$) 
components. For any code $C$, let $X^C$ be a random codeword 
picked uniformly from the code. For any set $J \subset [1:n]$, let 
$X^C|_J$ denote the components of $X^C$ in $J$.
$H(X^C|_J)$ is the entropy of the vector $X^C|_J$.
This is a function of the code, and so $H(X^{\cal C}|_J)$ is a random variable under the
above random code construction. The following lemma states that
for any $r<r^\prime$ and $\beta<r^\prime -r$, if $|J|\geq nr$, then with 
high probability over the code, $H(X^{\cal C}|_J)\geq nr - 2^{-n\beta}$. 

\begin{lemma}
\label{lem:random_coding_lemma}
Let $r < r' < 1$, and ${\cal C}$ be a random ($n, nr'$) binary code generated 
with i.i.d. Ber($\frac{1}{2}$) components. Let $X^{\cal C}|_J$ be as defined 
above. Then for any $\beta$, $0 < \beta < (r' - r)$, with high probability, the code satisfies the following property: \hspace*{2mm}
for any set $J \subset [1:n]$ with $|J| \geq nr$, $H(X^{\cal C}|_J) \geq nr - 2^{-n\beta}$. 
\end{lemma}

\begin{proof}
Wlog, let us assume $|J|=nr$. Let the $N=2^{nr^\prime}$ random codewords be 
denoted by $X_i;\, i=1,2,\cdots, N$. Let us denote the codeword components
in $J$ by $Y_i := X_i|_J$. By assumption, 
clearly $Y_i$ are independent and uniformly distributed
over the $2^{nr}$ binary strings in $\{0,1\}^{nr}$. Let us denote
the empirical distribution of $Y_i; \, i=1,2,\cdots, N$ as $\hat{p}_J$.
We are interested in the entropy $H(\hat{p}_J)$.

By Sanov's theorem, 
\begin{align*}
Pr \left[H(\hat{p}_J) < nr - 2^{-n\beta}\right] & \leq (N+1)^{2^{nr}} 2^{-ND(p^*||u)}
\end{align*}
where $u$ dentoes the uniform distribution over $\{0,1\}^{nr}$, and
\begin{align*}
p^* = \arg\min_{p:H(p) < nr - 2^{-n\beta}} D(p||u). &
\end{align*}
Clearly,
\begin{align*}
D(p^*||u) = nr - H(p^*) > 2^{-n\beta}. &
\end{align*}
So
\begin{align*}
Pr \left[H(\hat{p}_J)< nr - 2^{-n\beta}\right] & < (2^{nr^\prime} +1)^{2^{nr}} 2^{-2^{nr^\prime}\cdot 2^{-n\beta}} \\
& < 2^{n(r^\prime +1/n)\cdot 2^{nr}} \cdot 2^{-2^{nr^\prime}\cdot 2^{-n\beta}} \\
& \leq 2^{-2^{nr}(2^{n(r^\prime -r -\beta)}-n(r^\prime +1/n))}.
\end{align*}
By union bound,
\begin{align*}
& Pr \left[H(\hat{p}_J)< nr - 2^{-n\beta} \text{ for some } J\right]   \\
& \hspace*{10mm} \leq {n \choose nr} 2^{-2^{nr}(2^{n(r^\prime -r -\beta)}-n(r^\prime +1/n))}\\
& \hspace*{10mm} \leq 2^{nH(r)}\cdot 2^{-2^{nr}(2^{n(r^\prime -r -\beta)}-n(r^\prime +1/n))}.
\end{align*}
Since $\beta < r^\prime -r$, the above upper bound goes to zero as $n \longrightarrow \infty$.

\end{proof}

\end{document}